\documentclass[12pt, leqno, letterpaper]{amsart}

\usepackage{graphicx}    
\usepackage{bbm, xcolor}
\usepackage{xfrac}
\usepackage{multicol}
\usepackage{soul}

\usepackage[margin=1in]{geometry}

\usepackage[english]{babel}

\usepackage{amsmath,amsthm,amsfonts,amssymb}

\DeclareMathOperator\arctanh{arctanh}
\DeclareMathOperator\sgn{sgn}

\newtheorem{theorem}{Theorem}[section]

\newtheorem{lemma}{Lemma}[section]
\newtheorem{proposition}{Proposition}[section]
\newtheorem{assumption}{Assumption}[section]
\newtheorem{definition}{Definition}[section]
\newtheorem{remark}{Remark}[section]
\numberwithin{figure}{section}

\begin{document}

\title[Short-maturity asymptotics for option prices]
{Short-maturity asymptotics for option prices with interest rates effects}

\author{Dan Pirjol}
\address{Stevens Institute of Technology, Hoboken, New Jersey, United States of America}

\author{Lingjiong Zhu}
\address{Florida State University, Tallahassee, Florida, United States of America}

\date{February 21, 2024}


\begin{abstract}
We derive the short-maturity asymptotics for option prices in the local volatility model
in a new short-maturity limit $T\to 0$ at fixed $\rho = (r-q) T$, where $r$ is the interest rate and $q$ is the dividend yield. 
In cases of practical relevance $\rho$ is small, however our result holds for any fixed $\rho$.
The result is a generalization of the Berestycki-Busca-Florent formula \cite{BBF} for the short-maturity asymptotics of the implied volatility which includes interest rates and dividend yield effects
of $O(((r-q) T)^n)$ to all orders in $n$. We obtain analytical results for the ATM volatility and skew in this asymptotic limit. 
Explicit results are derived for the CEV model. 
The asymptotic result is tested numerically against exact evaluation in the
square-root model model $\sigma(S)=\sigma/\sqrt{S}$, which demonstrates that the
new asymptotic result is in very good agreement with exact evaluation in a wide range of model parameters relevant for practical applications.
\end{abstract}

\maketitle

\baselineskip18pt

\section{Introduction}

The simplest model for the risk-neutral dynamics of an asset price which is consistent with the observed market prices of the vanilla options with all strikes and maturities is the local volatility model \cite{Dupire1994,Derman1996}. 
This model is widely used in financial practice for 
pricing equities, FX and commodities derivatives.

Under the local volatility model the asset price $S_t$ is assumed to follow the process
under the risk-neutral probability measure $\mathbb{Q}$
\begin{equation}\label{SDE}
\frac{dS_t}{S_t} = \sigma(S_t) dW_t + (r - q) dt \,,\quad S_0 >0\,,
\end{equation}
where $W_t$ is a standard Brownian motion, $r$ is the risk-free rate, 
$q$ is the dividend yield and $\sigma(\cdot )$ is the local volatility function.
We assume that the local volatility function is a function of the asset price only and 
hence the local volatility model is time-homogeneous. 
The time homogeneity assumption can be relaxed in principle, to allow for a general time-dependent $\sigma(t,\cdot)$, 
although it has been shown that under mild conditions, even with this more general volatility function, 
short-maturity asymptotics only depends on $\sigma(0,\cdot)$; see e.g. \cite{BBF}.
It is therefore reasonable to assume that our results hold also for time-dependent volatility under mild conditions
with $\sigma(\cdot)$ being replaced by $\sigma(0,\cdot)$ to extend our results.

The short maturity limit of the implied volatility in the model (\ref{SDE}) was obtained
by Berestycki, Busca, Florent \cite{BBF}:
\begin{equation}
\lim_{T\to 0} \sigma_{BS}(K,T) = \frac{\log(K/S_0)}{\int_{S_0}^K \frac{dx}{x\sigma(x)}}\,,
\end{equation}
where $K>0$ is the strike price.
We refer to Lee \cite{Lee2004} for an overview of this formula and its properties. 
This result can also be obtained using large deviations theory \cite{Pham2007}.

The higher orders $O(T,T^2)$ in the short maturity expansion of the implied volatility in the local volatility model have been obtained using an expansion of the Dupire formula in \cite{HLbook} and using a heat kernel expansion by Gatheral et al. \cite{Gatheral}. 
A Taylor expansion for the implied volatility has been proved in \cite{Pag2017}. The asymptotics of the pricing PDE was 
studied in \cite{Cheng2011}.
Similar small maturity expansions have been obtained in stochastic volatility models 
\cite{SABR,Forde2009,Forde2012,Akahori2022}, local-stochastic volatility models \cite{Forde2011,Lorig2017} and models with jumps \cite{Alos2007,Figueroa2012}.

A generic feature of the leading order asymptotics as the maturity $T\to 0$ is the independence of the result on interest rates and dividend yields. Interest rates effects appear first at $O(T)$. The absence of these contributions introduces errors and reduces the practical usefulness of the leading order result. 
In this paper we present a modified short-maturity limit in the local volatility model, which includes contributions from interest rates effects 
already at the leading order in the expansion.

In this paper we consider asymptotics of option prices assuming that the asset price follows the model (\ref{SDE}) in the short-maturity limit 
\begin{equation}\label{rholimit}
T\to 0 \,,\quad \mbox{ at fixed } \rho = (r - q) T \,,
\end{equation}
where $rT^{2}=o(1)$ as $T\rightarrow 0$.
In most practical applications the $\rho$ parameter is small in absolute value. For example, assuming $r=6\%$, $q=0$ and an option maturity of 6 months, we have $T=0.5$ and $\rho=0.03$. However, theoretically, our results 
do not rely on the smallness of the $\rho$ parameter and hold for any fixed $\rho$.

We illustrate the effect of this limit by taking in (\ref{SDE}) the substitution $t\to \tau T$ with 
$\tau \in [0,1]$ and $(r-q) \to \frac{\rho}{T}$, such that $(r-q) T$ is fixed.
With this substitution the diffusion (\ref{SDE}) becomes
\begin{equation}
\frac{dS_{\tau T}}{S_{\tau T}} = 
\sigma(S_{\tau T}) dW_{\tau T} + T (r - q) d\tau 
= \sqrt{T} \sigma(S_{\tau T}) dW_{\tau} + \rho d\tau \,.
\end{equation}
As $T\to 0$ the volatility term is a small perturbation, but the drift term is constant and has to be fully taken into account.
This can be contrasted with the usual small-maturity limit where $r-q$ is kept fixed. In this case the drift term is $O(T)$ and does not contribute at leading order as $T\to 0$.

The limit (\ref{rholimit}) was previously considered in \cite{PZDiscreteAsian} for obtaining short-maturity
asymptotics for Asian options in the Black-Scholes model with non-zero $r$ and $q$.
Numerical results showed that including the $\rho$ dependence improved considerably
the agreement with exact benchmark evaluations, compared with the simple $T\to 0$
asymptotic result. A similar limit was used recently in \cite{PZLaplace} 
to obtain asymptotics for the Laplace transform of the time integral of the geometric Brownian motion.

The paper is organized as follows.
In Section~\ref{sec:main}, we derive the short-maturity asymptotics for out-of-the-money European options in the local volatility model in the limit (\ref{rholimit})
using large deviations theory. 
The result is expressed in terms of a rate function which is expressed as an integral of a functional along an optimal path, joining the spot price with the option strike, and determines the asymptotic implied volatility. The optimal path and the rate function are determined by the solution of a variational problem.
In Section~\ref{sec:variational}, we further analyze and solve this variational problem.
We give closed form results for the asymptotic implied volatility, its at-the-money (ATM) level and skew. 
In contrast to the usual short-maturity asymptotic result \cite{BBF}, under our limit the ATM level and skew of the implied volatility depend on an average of the local volatility in a region of log-strikes around the ATM point of width $\sim (r-q)T$.
Furthermore we show how our result reproduces the known result in the literature for the $O(rT)$ contribution to the implied volatility \cite{HLbook,Gatheral}, 
when expanded in $\rho$ to $O(\rho)$.
In Section~\ref{sec:CEV}, we apply our asymptotic results to the
CEV model and provide numerical experiments that show good performance of our method.
Apart from the theoretical interest,
including interest rates effects in option pricing
should be also of practical relevance, especially in the current economic environment of increasing interest rates. 
A few Appendices give background for large deviations theory and technical details and proofs of the results in Section \ref{sec:CEV}.

\section{Main Result}\label{sec:main}

We assume that the local volatility function $\sigma(x)$ in (\ref{SDE})
satisfies the following assumption.

\begin{assumption}\label{assump:main}
$\sigma(x)$ is bounded, i.e. $0<M_L \leq \sigma(\cdot) \leq 
M_U<\infty$, is differentiable, and satisfies a H\"{o}lder condition 
$|\sigma(e^x) - \sigma(e^y)| \leq M |x-y|^\eta$, with some $M,\eta>0$,
for any $x,y$.
\end{assumption}

The boundedness and H\"{o}lder conditions are not satisfied by some models that are popular in financial practice such as the CEV model $\sigma(S)=\sigma S^{\alpha-1}$. In Section~\ref{sec:CEV}, we will discuss how to relax Assumption~\ref{assump:main} to extend our analysis to include the CEV model.


European call and put option prices are given by risk-neutral expectations
\begin{equation}
C(K,T) = e^{-rT} \mathbb{E}[(S_T - K)^+] \,,\quad
P(K,T) = e^{-rT} \mathbb{E}[(K - S_T )^+] \,,
\end{equation}
where $K>0$ is the strike price and $S_{T}$ is the asset price at maturity $T>0$
with $x^{+}$ denoting $\max(x,0)$ for any $x\in\mathbb{R}$.
The forward price for maturity $T$ is $F(T) = S_0 e^{(r-q)T}$. 
Call options with $K> F(T)$ are out-of-the-money (OTM), with 
$K < F(T)$ are in-the-money (ITM) and with $K=F(T)$ are at-the-money (ATM). 

\begin{theorem}\label{thm:LD}
Assume that the asset price $S(t)$ follows the local volatility model
(\ref{SDE}) and Assumption~\ref{assump:main} is satisfied.
The asymptotics of call and put options in the short-maturity limit $T\to 0$
at fixed $\rho=( r - q) T$ with $rT^{2}=o(1)$ satisfy
\begin{eqnarray}\label{limC}
&& \lim_{T\to 0} T \log C(K,T) = - I(K,S_0) \,,\qquad K > S_0\,, \\
\label{limP}
&& \lim_{T\to 0} T \log P(K,T) = - I(K,S_0) \,,\qquad K < S_0 \,,
\end{eqnarray}
where the rate function $I(K,S_0)$ is given by
\begin{equation}\label{Isol}
I(K,S_0) = \inf_{g \in \mathcal{G}} \frac12 \int_0^1 \left( \frac{g'(t) - \rho}{\sigma(S_0 e^{g(t)})}
\right)^2 dt\,,
\end{equation}
where
\begin{equation}
\mathcal{G} := \left\{ g | g(0) = 0, g(1) = \log\left(K/S_{0}\right)\,, g \in \mathcal{AC}_{[0,1]} \right\}
\end{equation}
with $\mathcal{AC}_{[0,1]}$ being the set of all absolutely continuous functions on $[0,1]$,
and $\infty$ otherwise.
\end{theorem}

\begin{proof}
Since $rT^{2}=o(1)$ as $T\rightarrow 0$, 
we have
\begin{equation*}
\lim_{T\to 0} T \log C(K,T) =
\lim_{T\to 0} T \log e^{-rT}\mathbb{E}[(S_{T}-K)^{+}]
=\lim_{T\to 0} T \log \mathbb{E}[(S_{T}-K)^{+}].
\end{equation*}
Using standard arguments, see e.g. \cite{Pham2007}, one can show that the small-time asymptotics of the call option price with $K>S_{0}$ is related to the small-time asymptotics of the density of the asset price in the right wing
\begin{equation}\label{limT}
\lim_{T\to 0} T \log \mathbb{E}[(S_{T}-K)^{+}] = \lim_{T\to 0}
T \log \mathbb{Q}(S_T \geq K)\,,\quad K > S_0.
\end{equation}
A similar relation holds between the small-time asymptotics of the put options 
and of the density of $S_T$ in the left wing ($K < S_0$).
For both cases the limit (\ref{limT}) can be computed using large deviations theory as follows. 

Denoting $X_{t}:=\log S_{t}$, we have the stochastic differential equation
\begin{equation}\label{original:SDE}
dX_{t}=\sigma(e^{X_{t}})dW_{t}+\left(r - q -\frac{1}{2}\sigma^{2}\left(e^{X_{t}}\right)\right)dt.
\end{equation}
We are interested in the asymptotics in the limit $T\rightarrow 0$
with $(r-q)T=\rho$ being a fixed constant. 
Let us define a new probability measure $\hat{\mathbb{Q}}$ via
the Radon-Nikodym derivative:
\begin{equation}
\frac{d\mathbb{Q}}{d\hat{\mathbb{Q}}}\bigg|_{\mathcal{F}_{T}}
=e^{\int_{0}^{T}\frac{r - q}{\sigma(e^{X_{t}})}d\hat{W}_{t}
-\frac{1}{2}\int_{0}^{T}\frac{(r - q)^{2}}{\sigma^{2}(e^{X_{t}})}dt},
\end{equation}
where by Girsanov theorem, 
\begin{equation}
\hat{W}_{t}:=W_{t}+\int_{0}^{t}\frac{r - q}{\sigma(e^{X_{s}})}ds
\end{equation}
is a standard Brownian motion under the $\hat{\mathbb{Q}}$ measure, 
and we can rewrite \eqref{original:SDE} as
\begin{equation}\label{hat:SDE}
dX_{t}=\sigma(e^{X_{t}})d\hat{W}_{t}-\frac{1}{2}\sigma^{2}(e^{X_{t}})dt.
\end{equation}

A sufficient condition that the Radon-Nikodym derivative is a martingale is given by
the Novikov condition, which requires that
$\mathbb{E}\left[e^{\frac{1}{2}\int_{0}^{T}\frac{(r-q)^{2}dt}{\sigma^2(e^{X_{t}})}  } \right] < \infty$, where $T \geq 0$.
This condition is satisfied if the local volatility function $\sigma(\cdot)$ is strictly positive and is bounded
from below as $\sigma(\cdot ) \geq M_L >0$ under Assumption~\ref{assump:main}. 

It is known \cite{Varadhan} that under Assumption~\ref{assump:main} 
$\hat{\mathbb{Q}}(X_{\cdot T}\in\cdot)$
satisfies a sample path large deviation principle with rate function
\begin{equation}
I(g) = \frac{1}{2}\int_{0}^{1}\left(\frac{g'(t)}{\sigma(e^{g(t)})}\right)^{2}dt
\end{equation}
with $g(0)=\log S_0$ and $g\in \mathcal{AC}[0,1]$ being the set of all absolutely continuous functions on $[0,1]$, and $I(g)=\infty$ otherwise. See Appendix~\ref{app:LD} for a formal definition of the large deviation principle.

For call options with $K>S_{0}$, we have
\begin{align*}
\mathbb{Q}(S_{T}\geq K)
=\mathbb{E}[1_{X_{T}\geq \log K}]
=\hat{\mathbb{E}}\left[e^{\int_{0}^{T}\frac{r - q}{\sigma(e^{X_{t}})}d\hat{W}_{t}
-\frac{1}{2}\int_{0}^{T}\frac{(r - q)^{2}}{\sigma^{2}(e^{X_{t}})}dt}
\cdot 1_{X_{T}\geq \log K}\right].
\end{align*}
On the other hand, by dividing both hand sides of \eqref{hat:SDE}
by $\sigma^{2}(e^{X_{t}})$, we obtain
\begin{equation}\label{hat2:SDE}
\frac{dX_{t}}{\sigma^{2}(e^{X_{t}})}=\frac{1}{\sigma(e^{X_{t}})}d\hat{W}_{t}-\frac{1}{2}dt,
\end{equation}
so that
\begin{align*}
\mathbb{Q}(S_{T}\geq K)
&=\hat{\mathbb{E}}\left[e^{\int_{0}^{T}\frac{r - q}{\sigma^{2}(e^{X_{t}})}dX_{t}
+\frac{1}{2}(r - q)T
-\frac{1}{2}\int_{0}^{T}\frac{(r - q)^{2}}{\sigma^{2}(e^{X_{t}})}dt}
\cdot 1_{X_{T}\geq \log K}\right]
\\
&=e^{\frac12 \rho}\cdot\hat{\mathbb{E}}
\left[e^{\int_{0}^{T}\frac{r-q}{\sigma^{2}(e^{X_{t}})}dX_{t}
-\frac{1}{2}\int_{0}^{T}\frac{(r - q)^{2}}{\sigma^{2}(e^{X_{t}})}dt}
\cdot 1_{X_{T}\geq \log K}\right].
\end{align*}
By applying Varadhan's lemma (see Appendix~\ref{app:LD} for the precise statement), we obtain
\begin{align*}
&\lim_{T\rightarrow 0}T\log\hat{\mathbb{E}}\left[e^{\int_{0}^{T}\frac{r-q}{\sigma^{2}(e^{X_{t}})}dX_{t}
-\frac{1}{2}\int_{0}^{T}\frac{(r-q)^{2}}{\sigma^{2}(e^{X_{t}})}dt}
\cdot 1_{X_{T}\geq \log K}\right]
\\
&=\lim_{T\rightarrow 0}T\log\hat{\mathbb{E}}\left[e^{\frac{1}{T}\int_{0}^{1}\frac{\rho}{\sigma^{2}(e^{X_{tT}})}dX_{tT}
-\frac{1}{2}\frac{1}{T}\int_{0}^{1}\frac{\rho^{2}}{\sigma^{2}(e^{X_{tT}})}dt}
\cdot 1_{X_{T}\geq \log K}\right]
\\
&=\sup_{g\in\mathcal{AC}[0,1]:g(0)=\log S_{0}, g(1)\geq\log K}\left\{\int_{0}^{1}\frac{\rho g'(t)}{\sigma^{2}(e^{g(t)})}dt
-\frac{1}{2}\int_{0}^{1}\frac{\rho^{2}}{\sigma^{2}(e^{g(t)})}dt
-\frac{1}{2}\int_{0}^{1}\left(\frac{g'(t)}{\sigma(e^{g(t)})}\right)^{2}dt\right\}
\\
&=-\inf_{g\in\mathcal{AC}[0,1]:g(0)=\log S_{0},g(1)\geq\log K}
\frac{1}{2}\int_{0}^{1}\left(\frac{g'(t)-\rho}{\sigma(e^{g(t)})}\right)^{2}dt.
\end{align*}

It is convenient to subtract the value $\log S_0$ by redefining 
$g(t) \to g(t) - \log S_0$. This yields the stated result (\ref{Isol}) for the rate function.
The case of the puts with $K<S_{0}$ is obtained in a similar way.
The proof is complete.
\end{proof}

\begin{remark}
In the special case of the Black-Scholes model where $\sigma(\cdot ) = \sigma$ is a constant,
the rate function is 
$I(K,S_0) = \frac{1}{2\sigma^2} \Big( \log\frac{K}{S_0} - \rho \Big)^2 = 
\frac{1}{2\sigma^2} x^2$
where $x = \log\frac{K}{F(T)}$ denotes the log-moneyness of the option 
and $F(T) =S_0 e^\rho$ is the forward price. It is easy to check that
this agrees with the Black-Scholes formula under our asymptotic regime.
\end{remark}

\begin{remark}
The asymptotics \eqref{limC}-\eqref{limP} in Theorem~\ref{thm:LD} are for call options with $K>S_{0}$ and put options
with $K<S_{0}$.
One can simply apply put-call parity to obtain the corresponding asymptotics for put options with $K>S_{0}$ and call options
with $K<S_{0}$.
\end{remark}

For $\rho=0$, the variational problem \eqref{Isol} of Theorem~\ref{thm:LD} simplifies 
when expressed in terms of the function
\begin{equation}
h(t) = y( g(t) - \log S_{0})\,,
\end{equation}
with $y(x) := \int_0^x \frac{dz}{\sigma(S_{0} e^z)}$. Using $h'(t) = 
\frac{g'(t)}{\sigma(S_0 e^{g(t)})}$, the variational problem \eqref{Isol} becomes
\begin{equation}
I(K,S_0) = \inf_{h} \frac12 \int_0^1 [h'(t)]^2 dt\,,
\end{equation}
where $h(0)=0,h(1) = \log\frac{K}{S_0}$.
The solution for $h(t)$ satisfies the Euler-Lagrange equation $h''(t)=0$. 
The solution is a linear function of the form $h(t)=y(x) t$, and the rate function can be found in closed form
\begin{equation}\label{Irho0}
    I(K,S_{0}) = \frac12 y^2(x)\,, \quad (\rho=0) \,.
\end{equation}

The short-maturity asymptotics of Theorem~\ref{thm:LD} can be formulated as a
short-maturity limit for the implied volatility 
\begin{equation}
\lim_{T\to 0} \sigma_{BS}^2(K,T) = \frac{(\log \frac{K}{S_0} - \rho)^2}{2I(K,S_0)}
:= \sigma_{\mathrm{BBF},\rho}(K,S_0;\rho) \,.
\end{equation}

The $\rho=0$ limiting result (\ref{Irho0}) gives the short-maturity implied volatility
\begin{equation}\label{sigBBF}
\lim_{T\to 0} \sigma_{BS}(K,T) = \frac{\log \frac{K}{S_0}}{\int_{S_0}^K \frac{dx}{x\sigma(x)}} := \sigma_{\mathrm{BBF}}(K,S_0)\,,
\end{equation}
which recovers the well-known BBF formula \cite{BBF} for the leading
short-maturity asymptotics of 
the implied volatility in the local volatility model.

We study next the solution of the variational problem for $\rho\neq 0$,
and give an explicit result for the rate function $I(K,S_0)$.

\section{Solution of the Variational Problem}\label{sec:variational}

We give in this section the solution of the variational problem \eqref{Isol} in Theorem~\ref{thm:LD} with non-zero $\rho$. 
We start by studying the properties of the optimizer
$g(t)$ in this variational problem and classify the solutions of the Euler-Lagrange 
equation into three distinct classes. 
Then we give an explicit result for the rate function $I(K,S_0)$ in terms of quadratures.

\begin{proposition}
The optimizer $g(t)$ in the variational problem of Theorem
\ref{thm:LD} satisfies the Euler-Lagrange equation
\begin{equation}\label{ELg}
g''(t) = S_0 e^{g(t)} \frac{\sigma'(S_0 e^{g(t)}) }{\sigma ( S_0 e^{g(t)} ) } 
\left[ (g'(t))^2 - \rho^2\right]\,,
\end{equation}
with boundary conditions $g(0)=0$ and $g(1) = \log\frac{K}{S_0}$.
\end{proposition}

\begin{proof}
Define
\begin{equation}
L(g(t),g'(t)):=\frac{1}{2}\left(\frac{g'(t)-\rho}{\sigma(S_0 e^{g(t)})}\right)^{2}\,.
\end{equation}
The Euler-Lagrange equation for the variational problem \eqref{Isol} reads
\begin{equation}
\frac{\delta L}{\delta g}=\frac{d}{dt}\frac{\delta L}{\delta g'} \,.
\end{equation}

This gives
\begin{align}
\left(g'(t)-\rho\right)^{2}\frac{-S_0 \sigma'(S_0 e^{g(t)})e^{g(t)}}{\sigma^{3}(S_0 e^{g(t)})}
&=\frac{d}{dt}\left(\frac{g'(t)-\rho}{\sigma^{2}(S_0 e^{g(t)})}\right)
\nonumber
\\
&=\frac{g''(t)}{\sigma^{2}(S_0 e^{g(t)})}
+2\left(g'(t)-\rho\right)g'(t)\frac{-S_0 \sigma'(S_0 e^{g(t)})e^{g(t)}}{\sigma^{3}(S_0 e^{g(t)})}\,,
\end{align}
which is equivalent with the equation (\ref{ELg}).
This completes the proof.
\end{proof}

The solutions of the Euler-Lagrange equation (\ref{ELg}) have a constant of motion.

\begin{proposition}
Assume that $g(t)$ satisfies the Euler-Lagrange equation (\ref{ELg}).
Then 
\begin{equation}\label{Cdef}
C := \frac{1}{\sigma^2(S_0 e^{g(t)} )} \left[ (g'(t))^2 - \rho^2\right]
\end{equation}
is a constant.
\end{proposition}

\begin{proof}
The result follows by explicit computation of the derivative
\begin{equation}
\frac{d}{dt} \left( \frac{1}{\sigma^2(S_0 e^{g(t)} )} \left[ (g'(t))^2 - \rho^2\right]\right) = 0\,.
\end{equation}
Substituting here the equation (\ref{ELg}) yields the result shown.
\end{proof}

The constant $C$ is related to the derivative $g'(0)$ as
\begin{equation}
C = \frac{1}{\sigma^2(S_0)} [(g'(0))^2 - \rho^2]\,.
\end{equation}
This implies that the range of possible values for $C$ is
\begin{equation}
C \in \Big[-\frac{\rho^2}{\sigma^2(S_0)} , \infty\Big) \,.
\end{equation}
The minimal value of this constant corresponds to the trajectory with $g'(0)=0$.

The conservation of the quantity $C$ along each solution of the Euler-Lagrange equation 
can be used to classify the solutions of this equation into several groups, based on the sign of $C$. 

\subsection{Trajectories classification}
\label{sec:3.1}

The optimal trajectories $g(t)$ can be classified into 3 distinct classes.
Recall that any trajectory joins the origin $g(0)=0$ with
$g(1) = \log\frac{K}{S_0}$.

Define the three regions (see Figure~\ref{Fig:1} for a graphical representation):
\begin{enumerate}

\item Region 1: $g(t) \geq |\rho | t$. This region contains trajectories with
$K \geq S_0 e^{|\rho |}$. For either sign of $\rho$ this corresponds to OTM call options.

\item Region 2: $g(t) \leq -|\rho | t$. This region contains trajectories with 
$K \leq S_0 e^{-|\rho |}$. For either sign of $\rho$ this corresponds to OTM put options.

\item Region 3: $-|\rho | t < g(t) < |\rho | t$. This region contains trajectories with 
$S_0 e^{-| \rho |} < K < S_0 e^{| \rho |}$. For $\rho>0$ this region corresponds to
OTM put options, and for $\rho <0$ to OTM call options. 

\end{enumerate}

Consider first a trajectory with $C>0$. 
This has either $g'(0)> |\rho |$ or $g'(0) < - |\rho |$.
By continuity of $g'(t)$ on $t\in [0,1 ]$, the same inequalities are preserved for all $t\in [0,1]$. 
Such a trajectory can belong either to the region denoted 1 or 2. 

Any trajectory with $C<0$ has $-|\rho | < g'(t) < |\rho |$. From  $g(t) = \int_0^t g'(s) ds$ we have
$| g(t) | \leq \int_0^t |g'(s)| ds \leq \rho t$, which implies that the trajectory is contained in the triangular region 3. 

In regions 1 and 2 the derivative $g'(t)$ is bounded from below (above) by 
$|\rho |$ ($-|\rho |$), so it can never vanish.
On the other hand, the slope $g'(t)$ of a trajectory in region 3 may vanish at some point $t_*\in [0,1]$ and change sign.

\begin{figure}[h]
\centering
\includegraphics[width=8cm]{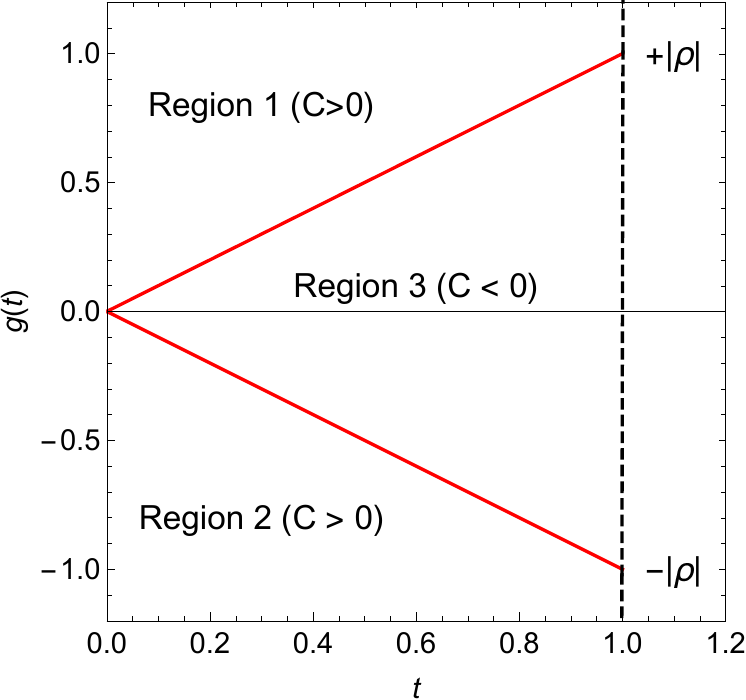}
\caption{The optimal trajectories $g(t)$ may belong to one of three regions, with definite signs of $C$ as shown.}
\label{Fig:1}
\end{figure}

We study the determination of the constant $C$ for each type of trajectory. 

i) Trajectories in region 1. For this case $g'(t)\geq |\rho | > 0$ is positive.
The function $g(t)$ is monotonically increasing with slope
\begin{equation}
g'(t)=\sqrt{C \sigma^{2}(e^{g(t)})+\rho^{2}} \,.
\end{equation}

Integrating over $t\in [0,1]$ we have
\begin{equation}\label{eqC1}
1=\int_{0}^{1}\frac{g'(t)dt}{\sqrt{C \sigma^{2}(e^{g(t)})+\rho^{2}}}
=\int_{0}^{\log(K/S_0)}\frac{dx}{\sqrt{C \sigma^{2}(S_0 e^x)+\rho^{2}}}.
\end{equation}
This uniquely determines the constant $C$ for given $(S_0,K)$.

ii) Trajectories in region 2. For this case $g'(t)\leq- |\rho | < 0$ is negative. 
$g(t)$ is monotonically decreasing, with slope
\begin{equation}
g'(t)=- \sqrt{C \sigma^{2}(e^{g(t)})+\rho^{2}} \,.
\end{equation}

Integration gives
\begin{equation}\label{eqC2}
1=\int_{0}^{1}\frac{-g'(t)dt}{\sqrt{ C \sigma^{2}(e^{g(t)})+\rho^{2}}}
=\int_{\log(K/S_0)}^{0}\frac{dx}{\sqrt{C \sigma^{2}(S_0 e^x)+\rho^{2}}}.
\end{equation}
This uniquely determines the constant $C$.

iii) Trajectories in region 3. 
As mentioned above, for this case $g'(t)$ can change sign.
Let us study first the possible number of sign changes.
The following result shows that convexity can restrict the number of sign changes.

\begin{proposition}
Assume that $\sigma(S)$ is monotonically decreasing (increasing) in the region $S_0 e^{-|\rho |} < S < S_0 e^{|\rho |}$.
Then $g'(t)$ may have at most one zero.
\end{proposition}

\begin{proof}
The Euler-Lagrange equation gives
\begin{equation}
g''(t) = C \sigma(e^{g(t)} ) \sigma'(e^{g(t)}) e^{g(t)} \,.
\end{equation}

In region 3 we have $C<0$. 
Thus, if $\sigma(S)$ is decreasing (increasing) in this region, then $g(t)$ is convex (concave). 
This implies that $g'(t)$ may have at most one zero. 
\end{proof}

We assume in the following that $\sigma(S)$ is monotonic in the region 
$S_0 e^{-| \rho |} < S < S_0 e^{| \rho |}$.
For given $K\in (S_0 e^{-|\rho |}, S_0 e^{|\rho |})$, 
two types of solutions are possible: 

a) $g'(0)$ and $g'(1)$ have the same sign. Then $g'(t)$ does not have a sign change, and the equation for $C$ reduces to
(\ref{eqC1}) (for $g'(t)>0$) or (\ref{eqC2}) (for $g'(t)<0$).

b) $g'(0)$ and $g'(1)$ have opposite signs. Thus $g'(t)=0$ at some point $t_*$. For definiteness assume $g'(0)<0, g'(1)<0$.
Then the equation for $C$ reads
\begin{equation}
1=\int_0^{t_*} \frac{g'(t)dt}{\sqrt{C \sigma^{2}(e^{g(t)})+\rho^{2}}}
+ \int_{t_*}^1 \frac{-g'(t)dt}{\sqrt{C \sigma^{2}(e^{g(t)})+\rho^{2}}}\,,
\end{equation}
or equivalently
\begin{equation}\label{eqC3}
1 
=\int_{u_*}^0\frac{du}{\sqrt{C \sigma^{2}(S_0 e^u)+\rho^{2}}} +
\int_{u_*}^{\log(K/S_0)} \frac{du}{\sqrt{C \sigma^{2}(S_0 e^u)+\rho^{2}}}\,,
\end{equation}
with $u_* := g(t_*)$.

\subsection{Solution for the rate function}

The rate function of Theorem~\ref{thm:LD} can be expressed in terms of quadratures, separately in each of the three regions introduced above. 

\begin{proposition}\label{prop:I}
The rate function $I(K,S_0)$ is given by the following result. 

i) For $K \geq S_0 e^{|\rho |}$ (region 1), we have
\begin{equation}\label{Isol1}
I(K,S_0) = \frac12 \int_0^{\log(K/S_0)} 
\frac{(\sqrt{C \sigma^2(S_0 e^{g}) + \rho^2}-\rho)^2}
{\sigma^2(S_0 e^{g}) \sqrt{C \sigma^2(S_0 e^{g}) + \rho^2}   } dg\,,
\end{equation}
where $C$ is found as the solution of (\ref{eqC1}).

ii) For $0 < K \leq S_0 e^{-|\rho |}$ (region 2), we have
\begin{equation}\label{Isol2}
I(K,S_0) = \frac12 \int_{\log(K/S_0)}^0
\frac{(\sqrt{C \sigma^2(S_0 e^{g}) + \rho^2} + \rho)^2}
{\sigma^2(S_0 e^{g}) \sqrt{C \sigma^2(S_0 e^{g}) + \rho^2}   } dg\,,
\end{equation}
where $C$ is found as the solution of (\ref{eqC2}).

iii) For $S_0 e^{-|\rho |} < K < S_0 e^{| \rho |}$ (region 3), we distinguish further
between the two cases where

a) $g(t)$ is monotonic on $t\in [0,1]$.
For this case the rate function is given by (\ref{Isol1}) if $g(t)$ is increasing,
and by (\ref{Isol2}) if $g'(t)$ is decreasing.

b) $g'(0),g'(1)$ have opposite sign, and $g'(t)$ changes sign only once.
Take for definiteness $g'(0)<0,g'(1)>0$. Then the rate function is
\begin{equation}
I(K,S_0) = 
 \frac12 \int_{g_*}^0 
\frac{(\sqrt{C \sigma^2(S_0 e^{g}) + \rho^2}-\rho)^2}
{\sigma^2(S_0 e^{g}) \sqrt{C \sigma^2(S_0 e^{g}) + \rho^2}   } dg +
\frac12 \int_{g_*}^{\log(K/S_0)}
\frac{(\sqrt{C \sigma^2(S_0 e^{g}) + \rho^2} + \rho)^2}
{\sigma^2(S_0 e^{g}) \sqrt{C \sigma^2(S_0 e^{g}) + \rho^2}   } dg\,,
\end{equation}
where $C$ is found as the solution of (\ref{eqC3}) and $g_* = g(t_*)$ with
$g'(t_*)=0$.

\end{proposition}

\begin{remark} \label{rmk:1}
The rate function vanishes at $K=F(T) = S_0 e^\rho$, which corresponds to the 
at-the-money options. That is, 
\begin{equation}
I(K=S_0 e^\rho, S_0) = 0 \,.
\end{equation}
At this point the optimizer is $g(t) = \rho t$ and $C=0$. 

i) For $\rho > 0$ this point is the boundary of regions 1 and 3. 
The integrand in (\ref{Isol1}) vanishes, which gives $I(K=S_0 e^\rho,S_0)=0$.

ii) For $\rho <0$ this point is the boundary of regions 2 and 3. 
The integrand in (\ref{Isol2}) vanishes, which gives again $I(K=S_0 e^\rho,S_0)=0$.
\end{remark}

The asymptotic result of Theorem~\ref{thm:LD} 
is equivalent with a prediction for the short maturity 
limit of the implied volatility, which generalizes the BBF result (\ref{sigBBF}) to
non-zero $\rho$. We denote it as
\begin{equation}\label{sigBBFrho}
\lim_{T\to 0, r T = \rho} \sigma^2_{BS}(K,S_0,T) = \frac{(k - \rho)^2}{2I(K,S_0)} := \sigma_{\mathrm{BBF},\rho}(K,S_0,\rho)\,,
\end{equation}
where we denoted the log-strike as $k := \log\frac{K}{S_0}$. 
With non-zero $\rho$, this is different from the
log-moneyness $x = \frac{K}{F(T)} = k - \rho$.

\subsection{At-the-money implied volatility}

We give here an analytical result for the asymptotic implied volatility 
of an at-the-money option, including interest rates effects.

\begin{proposition}\label{prop:leading}
We have
\begin{equation}\label{sigRhoATM}
\sigma_{\mathrm{BBF},\rho}^2(K=S_0e^\rho, S_0) = 
\frac{1}{\rho} \int_0^\rho \sigma^2(S_0 e^u ) du \,.
\end{equation}
\end{proposition}

The asymptotic ATM implied variance with non-zero $\rho$ is an average of the 
local variance over a range of $S$ between the spot price $S_0$ and the forward
price $S_0e^\rho$. This is in contrast to the $\rho=0$ case, where the asymptotic
ATM implied volatility depends only on the spot local volatility. 

\begin{proof}[Proof of Proposition~\ref{prop:leading}]
As noted in Remark~\ref{rmk:1}, the rate function vanishes at $K=S_0 e^\rho$,
corresponding to the ATM point. Let us study the expansion of the rate function around 
this point. We will show that this expansion has the form
\begin{equation}
I(K,S_0) = a_0 x^2 + a_1 x^3 + O(x^4)\,,
\end{equation}
where $x = \log\frac{K}{S_0 e^\rho} $ is the option log-moneyness.

The ATM asymptotic implied volatility is determined by the coefficient $a_0$ as
\begin{equation}\label{eqATM}
\sigma_{\mathrm{BBF},\rho}^2(K=S_0e^\rho,S_0) = \frac{1}{2a_0} \,.
\end{equation}

The constant $C$ associated with the optimal path at $K=S_0 e^\rho$ vanishes.
It can be expanded in powers of the log-moneyness $x$  as
\begin{equation}\label{Cexp}
C(x) = c_0 x + c_1 x^2 + O(x^3)\,.
\end{equation}
Next, we determine the coefficient $c_0$. (The coefficient $c_1$ will
be derived below in (\ref{c1sol}).) 
We give the proof only for $\rho>0$
when $C$ is given by the solution of (\ref{eqC1}). The case $\rho<0$ is handled in a similar way.
The coefficients $c_j$ can be determined by taking derivatives of the relation (\ref{eqC1}) and taking $x\to 0$. At leading order this yields
\begin{equation}
\frac{1}{\rho} - \frac12 c_0 \int_0^\rho \frac{\sigma^2(S_0 e^u)}{\rho^3} du = 0\,,
\end{equation}
which gives
\begin{equation}\label{c0sol}
c_0 = \frac{2\rho^2}{\int_0^\rho \sigma^2(S_0 e^u) du } \,.
\end{equation}

Substituting the expansion (\ref{Cexp}) in the integrand of (\ref{Isol1}) 
and expanding in $x$ gives
\begin{equation}\label{Lexp}
\frac{(\sqrt{C(x) \sigma^2(S_0 e^{g}) + \rho^2}-\rho)^2}
{\sqrt{C(x) \sigma^2(S_0 e^{g}) + \rho^2}   } = \frac{c_0^2 \sigma^4(S_0 e^g)}{4\rho^3} x^2 + \frac{2 c_0 c_1 \rho^2 \sigma^2(S_0 e^g) - c_0^3 \sigma^4(S_0 e^g)}{8\rho^5} 
x^3 + O(x^4) \,.
\end{equation}
This gives the leading term in the expansion of the rate function 
\begin{equation}
I(K,S_0) = \frac{c_0^2}{8\rho^3} \int_0^\rho \sigma^2(S_0 e^u) du \cdot x^2 + O(x^3)\,,
\end{equation}
which yields
\begin{equation}
a_0 = \frac{c_0^2}{8\rho^3} \int_0^\rho \sigma^2(S_0 e^u) du \,.
\end{equation}
Substituting into (\ref{eqATM}) gives
\begin{equation}
\sigma_{\mathrm{BBF},\rho}^2(K=S_0e^\rho,S_0) = \frac{1}{2a_0} = \frac{4\rho^3}
{4\rho^4 } \int_0^\rho \sigma^2(S_0 e^u) du\,,
\end{equation}
which reproduces the stated result (\ref{sigRhoATM}).
\end{proof}

\subsection{The ATM implied volatility skew}

The ATM skew is defined as
\begin{equation}
s(T) = \frac{d}{dx} \sigma(K,S_0)|_{x=0} = K \frac{d}{dK} \sigma(K,S_0)|_{K=S_0}\,.
\end{equation}
It is well-known that the ATM skew of the short-maturity asymptotics of the implied volatility is one-half of the ATM skew of the local volatility, see e.g. \cite{Lee2004}.
\begin{equation}\label{onehalfrule}
K\frac{d}{dK} \sigma_{\mathrm{BBF}}(K,S_0)|_{K=S_0} = \frac12 S_{0} \frac{d}{dS}\sigma(S_0)\,.
\end{equation}
We present next the
generalization of this result to the short-maturity $T\to 0$ 
asymptotics of the implied volatility, taken at finite and fixed $(r-q)T = \rho$. 
In contrast to the result (\ref{onehalfrule}), under the small-$T$ limit at fixed $\rho$, the ATM skew depends on an weighted average of the local volatility in a range of values between spot $S_0$ and forward $S_0 e^\rho$.


\begin{proposition}
Denote $\sigma_{\mathrm{BBF},\rho}(ATM) := \sigma_{\mathrm{BBF},\rho}(K=S_0 e^\rho,S_0)$ which is given by (\ref{sigRhoATM}).

We have 
\begin{equation}\label{sATM}
\frac{1}{\sigma_{\mathrm{BBF},\rho}(ATM)}
\frac{d}{dx} \sigma_{\mathrm{BBF},\rho}(ATM) = 
- \frac12 \cdot
\frac{\int_0^\rho \sigma^2(S_0 e^u) [ \sigma^2(S_0 e^u) - \sigma^2(S_0 e^\rho) ] du}
{\left( \int_0^\rho \sigma^2(S_0 e^u) du \right)^2}\,.
\end{equation}
\end{proposition}

\begin{proof}
It is convenient to introduce the following notations
\begin{equation}
I_2(x) := \int_0^{\rho+x} \sigma^2(S_0 e^u) du\,, 
\qquad
I_4(x) := \int_0^{\rho+x} \sigma^4(S_0 e^u) du\,.
\end{equation}
We will denote the values of these integrals at the ATM point $x=0$ as
$I_{2,4} := I_{2,4}(0)$.

Using the same approach as in the proof of 
Proposition~\ref{prop:leading},
we obtain the coefficient of the $O(x^2)$ term in the expansion of $C(x)$ in \eqref{Cexp}
\begin{equation}\label{c1sol}
c_1 = -2\rho^2 \frac{\sigma^2(S_0 e^\rho)}{(I_2)^2} + 3\rho^2 \frac{I_4}{(I_2)^3}\,.
\end{equation}

The expansion of the rate function $I(K,S_0)$ in powers of $x$ is obtained by integrating (\ref{Lexp}). Substituting into this result the expressions for $c_0$ from 
(\ref{c0sol}) and $c_1$ from (\ref{c1sol}) we get the expansion to $O(x^3)$:
\begin{equation}
I(K,S_0) = \frac12 \rho \frac{1}{(I_2)^2} I_2(x) x^2 + 
 \left\{-\rho \frac{\sigma^2(S_0 e^\rho)}{(I_2)^2} + \frac12 \rho \frac{I_4}{(I_2)^3} \right\} x^3 + O(x^4)\,.
\end{equation}
Expanding further $I_2(x) = I_2 + \sigma^2(S_0 e^\rho) x + O(x^2)$ in the first term, 
gives
\begin{equation}
I(K,S_0) = \rho \frac{1}{2I_2} x^2 + 
\rho \cdot \frac{1}{2(I_2)^3} \left\{ I_4 - \sigma^2(S_0 e^\rho) I_2 \right\} x^3 + O(x^4)\,.
\end{equation}

Using the relation of the rate function to the asymptotic implied volatility yields the 
stated result for the $O(x)$ term in the asymptotic implied volatility.  
This completes the proof.
\end{proof}

\textit{Limiting case $\rho\to 0$.} We show that in the limit $\rho\to 0$,
we recover the result (\ref{onehalfrule}) for the short-maturity asymptotics of the skew
in the local volatility model in the absence of interest rates effects.
The $\rho\to 0$ limit of the ratio (\ref{sATM}) can be
evaluated using the L'H\^{o}spital rule:
\begin{equation}
\lim_{\rho\to 0}
\frac{\int_0^\rho \sigma^2(S_0 e^u) [ \sigma^2(S_0 e^u) - \sigma^2(S_0 e^\rho) ] du}
{\left( \int_0^\rho \sigma^2(S_0 e^u) du \right)^2}=
\lim_{\rho\to 0}
\frac{-\frac{d}{d\rho} \sigma^2(S_0 e^\rho)}{2\sigma^2 (S_0 e^\rho)}
= -S_0 \frac{\sigma'(S_0)}{\sigma(S_0)} \,.
\end{equation}
This gives the ATM skew:
\begin{equation}
\lim_{\rho\to 0} 
\frac{1}{\sigma_{\mathrm{BBF},\rho}(ATM)}
\frac{d}{dx} \sigma_{\mathrm{BBF},\rho}(ATM) = 
\frac12 S_0 \frac{\sigma'(S_0)}{\sigma(S_0)}\,, 
\end{equation}
which reproduces the well-known result (\ref{onehalfrule}).

\subsection{Leading $O(\rho)$ correction}

The $O(\rho)$ correction to the rate function can be obtained in closed form. 

\begin{proposition}\label{prop:I1}
The first two terms in the small $\rho$ expansion of the rate function are
\begin{equation}
I(K,S_0) = I_0(K, S_0) + \rho I_1(K, S_0) + O(\rho^2)\,,
\end{equation}
with
\begin{equation}\label{I1sol}
I_0(K,S_0) = \frac12 \left( \int_{S_0}^K \frac{du}{u \sigma(u) }\right)^2\,,
\qquad\text{and}\qquad
I_1(K, S_0) = - \int_{S_0}^K  \frac{du}{u \sigma^2(u) }\,. 
\end{equation}
\end{proposition}

\begin{proof}
Assume $\rho>0$ and $K>S_0 e^{|\rho|}$. 
Thus we use the result (\ref{Isol1}) for the rate function in Proposition \ref{prop:I} for $K$ in region 1. The result is the same for $\rho<0$ and for $K$ in all other regions. 

First expand the coefficient $C$ in powers of $\rho$, using (\ref{eqC1}).
The leading order term is
\begin{equation}\label{C0}
C=\left(\int_{S_{0}}^{K}\frac{dx}{x\sigma(x)}\right)^{2}+O(\rho^{2}).
\end{equation}

The rate function is expanded in $\rho$ as
\begin{align}
I(K,S_0) &= \frac{1}{2}\int_{S_{0}}^{K}\frac{C\sigma^{2}(x)+2\rho^{2}-2\rho\sqrt{C\sigma^{2}(x)+\rho^{2}}}{x\sigma^{2}(x)\sqrt{C\sigma^{2}(x)+\rho^{2}}}dx
\nonumber
\\
&=\frac{C}{2}\int_{S_{0}}^{K}\frac{dx}{x\sqrt{C\sigma^{2}(x)+\rho^{2}}}
-\rho\int_{S_{0}}^{K}\frac{dx}{x\sigma^{2}(x)}+O(\rho^{2})
\nonumber
\\
&=\frac{C}{2}-\rho\int_{S_{0}}^{K}\frac{dx}{x\sigma^{2}(x)}+O(\rho^{2})
\nonumber
\\
&=\frac{1}{2}\left(\int_{S_{0}}^{K}\frac{dx}{x\sigma(x)}\right)^{2}-\rho\int_{S_{0}}^{K}\frac{dx}{x\sigma^{2}(x)}+O(\rho^{2}),
\end{align}
where we used (\ref{C0}) in the last line. 
This completes the proof.
\end{proof}

This result is equivalent with a prediction for the $O((r-q)T)$ correction to the asymptotic implied volatility $\sigma_{BS}(K,S_0,T)$.
The $O(T)$ correction to the implied volatility in the local volatility model was computed by Henry-Labord\`ere \cite{HLbook}
and Gatheral et al. \cite{Gatheral}. 
Assuming an expansion of the form $\sigma_{BS}(K,S_0,T) = \sigma_0(K,S_0) + \sigma_1 (K, S_0) T + O(T^2)$, 
they find the following result for the $O((r-q)T)$ term 
(see equation (2.7) in \cite{Gatheral})
\begin{equation}\label{sig1G}
\sigma_1(K,S_0) = \frac{\sigma_0^3}{\log^2 \frac{K}{S_0} }
\left\{ (\cdots ) + (r-q) \int_{S_0}^K \left( \frac{1}{\sigma^2(u)} - \frac{1}{\sigma_0^2(K,S_0)} \right) du \right\}\,,
\end{equation}
where the ellipses denote terms independent of $r-q$.

We will show that the result (\ref{I1sol}) reproduces the correction term in (\ref{sig1G})
by expanding the asymptotic implied volatility $\sigma_{\mathrm{BBF},\rho}(K,S_0)$ defined in (\ref{sigBBFrho}) in powers of $\rho$
\begin{equation}
\sigma_{\mathrm{BBF},\rho}^2(K,S_0) = \frac{(k-\rho)^2}{2 I(K,S_0)} = \frac{k^2}{2I_0(K,S_0)} + \rho\left( - \frac{k}{I_0(K,S_0)} 
- \frac{k^2}{2I_0^2(K,S_0)} I_1(K,S_0) \right) + O(\rho^2)\,.
\end{equation}
We denoted here $k=\log(K/S_0)$ the log-strike, which is 
related to the log-moneyness $x$ as $x=k-\rho$.
Comparing with the expansion $\sigma_{BS}(K,S_0,T) = \sigma_0(K,S_0) + \sigma_1 (K, S_0) T + O(T^2)$,
this gives for the coefficient of $\rho$ in the $O(T)$ term 
\begin{equation}
\sigma_1(K,S_0)[\rho] = - \frac{1}{2\sigma_0} \left( \frac{k^2}{2I_0^2(K,S_0)} I_1(K,S_0) + \frac{k}{I_0(K,S_0)} \right)\,.
\end{equation}
Using $\frac{1}{I_0(K,S_0)} = \frac{2\sigma_0^2(K,S_0)}{k^2}$ and substituting the result 
for $I_1(K,S_0)$ from (\ref{I1sol}) gives
\begin{equation}
\sigma_1(K,S_0)[\rho] = - \frac{\sigma_0^3}{k^2} \left(  I_1(K,S_0) + \frac{k}{\sigma_0^2(K,S_0} \right)
= \frac{\sigma_0^3}{\log^2 \frac{K}{S_0} }
\left( \int_{S_0}^K \frac{du}{u \sigma^2(u) } - \frac{\log(K/S_0)}{\sigma_0^2(K,S_0)}
\right)\,,
\end{equation}
which is seen to coincide precisely with the coefficient of $r-q$ in (\ref{sig1G}).

Finally, we make a remark that the asymptotic result $\sigma_{\mathrm{BBF},\rho}(K,S_0)$ derived in this paper includes terms of order $O((r-q)T)^n)$ to all orders in $n$.


\section{Application: the CEV model}\label{sec:CEV}

\subsection{The CEV Model}
In this section, we consider the application of the asymptotic method to the CEV model 
\begin{equation}\label{CEVsde}
dS_t = \sigma S_t^\alpha dW_t + (r-q) S_t dt\,, 
\end{equation}
which corresponds to $\sigma(S) = \sigma S^{\beta}$ with $\beta=\alpha-1$. 
This model was first introduced by Cox \cite{Cox1975}. For a short survey
we refer to Linetsky and Mendoza \cite{Linetsky2010}. Closed form
option prices for this model have been obtained by Schroeder \cite{Schroeder}. 
This CEV model has leverage effect for $\beta<0$, which is the property that the
stock price volatility increases as the stock price decreases. For this reason we
will consider only the range $-\frac12\leq \beta < 0$.

The short-maturity limit $T\to 0$ of the implied volatility in this model is obtained from the BBF formula (\ref{sigBBF}) which gives
\begin{equation}\label{BBF:CEV}
\sigma_{\mathrm{BBF}}(K,S_0) = \sigma |\beta | \frac{\log(K/S_0)}{K^{-\beta} - S_0^{-\beta}}\,.
\end{equation}
This result does not depend on interest rates effects, which is a generic result for the leading order under the usual $T\to 0$ asymptotics. These effects appear first at $O(T)$.

In this section we present the asymptotic implied volatility $\sigma_{\mathrm{BBF},\rho}(K,S_0,\rho)$ for the CEV model under the modified short-maturity limit considered here $T\to 0$ at fixed $\rho=(r-q)T$. This is expressed as in (\ref{sigBBFrho}) in terms of a rate function $I(K,S_0)$ which is given for a general volatility function $\sigma(\cdot )$ in Theorem~\ref{thm:LD}. We evaluate this rate function explicitly for the CEV model in Proposition~\ref{prop:CEV} below.

First we need to address a technical point.
Our main result Theorem~\ref{thm:LD} (and hence the subsequent discussions in Section~\ref{sec:variational}) requires Assumption~\ref{assump:main}, which does 
not hold for the CEV model $\sigma(x) = \sigma x^{\alpha-1}$. However, we can extend Theorem~\ref{thm:LD} to cover the CEV model as follows. In the proof of Theorem~\ref{thm:LD}, Assumption~\ref{assump:main} 
is used for the short-maturity large deviations for diffusion processes and to check the Novikov condition.

For the large deviations for diffusion processes, a large deviations property for the square-root
process $\alpha=\frac12 $ was proved by Donati-Martin \textit{et al} in \cite{23}, 
which was generalized by Baldi and Caramelino \cite{Baldi} to a wider class of models,
including the CEV model with $1/2 \leq \alpha < 1$. We can use \cite{Baldi}
to replace the reference \cite{Varadhan} in the proof of Theorem~\ref{thm:LD}.
For the Novikov condition, a separate argument is required when the volatility function can vanish as in the CEV model. 
In Appendix \ref{app:1} we provide such a proof for the CEV model, and show that for sufficiently small $T$, the Novikov condition is satisfied.

\textbf{Notation.} 
In order to simplify the notation, in this section we denote $\rho \to \theta \rho$
with $|\rho | \to \rho >0$ the absolute value of the $\rho$ parameter
and $\theta = \sgn(\rho) = \pm 1$ the sign of this parameter. 
Many results depend only on the absolute value of $\rho$, and using $|\rho|$ would make the notation unnecessarily heavy.

\begin{proposition}\label{prop:CEV}
Assume that the asset price $S_t$ follows the CEV model (\ref{CEVsde}).

i) For $K > S_0 e^{\rho }$ (region 1) and $K < S_0 e^{-\rho }$ (region 2), the 
rate function is
\begin{equation}\label{ICEVsol1}
I(K,S_0) = 
\frac{S_0^{2|\beta|}}{|\beta | \sigma^2} \cdot 
\left( e^{|\beta| x} - 1\right)^2
\cdot
\begin{cases}
\frac{\rho}{1 - e^{-2\rho |\beta|}}\,,  & \theta = +1, \\ 
\frac{\rho}{e^{2\rho |\beta|} - 1}\,,   & \theta = -1, 
\end{cases}
\end{equation}
with $x := \log\frac{K}{S_0} - \rho \theta$ being the log-moneyness.

ii) For $S_0 e^{-\rho } < K < S_0 e^{\rho }$ (region 3), the rate function is
\begin{equation}\label{ICEVsol2}
I(K,S_0) = \frac{S_0^{2|\beta|}}{4 |\beta | \sigma^2 } \rho \left( 1- \frac{y_0^2}{\rho^2} \right) e^{-2 \arctanh ( y_0/\rho) }
\cdot
\begin{cases}
(1 - e^{-2|\beta | \rho})\,,  &  \theta = + 1, \\
(e^{2|\beta | \rho} - 1 )\,, &  \theta =  - 1, 
\end{cases}
\end{equation}
where 
\begin{equation}\label{y0sol2}
y_0 := \rho \frac{e^{|\beta | \log(K/S_0)} - \cosh( |\beta | \rho)}{\sinh ( |\beta | \rho) }\,.
\end{equation}
\end{proposition}

\begin{proof}
The proof is given in Appendix \ref{app:2}.
\end{proof}

\begin{remark}
Taking the $\rho\to 0$ limit, the rate function in Proposition \ref{prop:CEV} becomes
\begin{equation}
\lim_{\rho\to 0} I(K,S_0) = \frac{1}{2\beta^2 \sigma^2}
\Big( K^{|\beta |} - S_0^{|\beta |} \Big) \,.
\end{equation}
Substituting into (\ref{sigBBFrho}) this is seen to reproduce the BBF 
result for the CEV model (\ref{BBF:CEV}) under the usual $T\to 0$ limit.
\end{remark}

\begin{remark}
The asymptotic implied volatility at the ATM point is
\begin{equation}\label{sigATMrho}
\sigma_{\mathrm{BBF},\rho}^2(K=S_0 e^\rho,S_0) = \frac{\sigma^2}{S_0^{2|\beta |} } \cdot \frac{1 - e^{-2\rho |\beta |} }{2\rho | \beta |} \,.
\end{equation}
This follows from the general result (\ref{sigRhoATM}) using the volatility function
$\sigma(S) = \sigma S^{\beta}$. 

The asymptotic ATM normalized skew is obtained from (\ref{sATM}) with the result
\begin{equation}
\frac{1}{\sigma_{\mathrm{BBF},\rho}(K=S_0 e^\rho,S_0) }
\frac{d}{dx} \sigma_{\mathrm{BBF},\rho}(K=S_0 e^\rho,S_0) = - \frac12 |\beta |\,.
\end{equation}
The dependence on $\rho$ cancels out in the ratio between the ATM skew and the
ATM implied volatility. 

These results can be verified by expanding the closed form result for the rate function in powers of $x$. The ATM volatility is related to the coefficient of the $O(x^2)$ term, and the ATM skew is related to the coefficient of the $O(x^3)$ term. 
\end{remark}

\subsection{Numerical tests}

Analytical results for option prices in the CEV model are available from Schroeder (1989) \cite{Schroeder}. We will use them to test the numerical efficiency of the new asymptotic limit considered here, and compare it with the simple $T\to 0$ asymptotic limit.

For the numerical tests we take $\beta=-\frac12$ (square root model) and 
$S_0=2, \sigma=0.14$, similar to the first scenario of Dassios and Nagardjasarma \cite{Dassios}. This corresponds to ATM implied volatility close to 
$\frac{\sigma}{\sqrt{S_0}} = 0.1$. 

\begin{figure}[h]
\centering
\includegraphics[width=6cm]{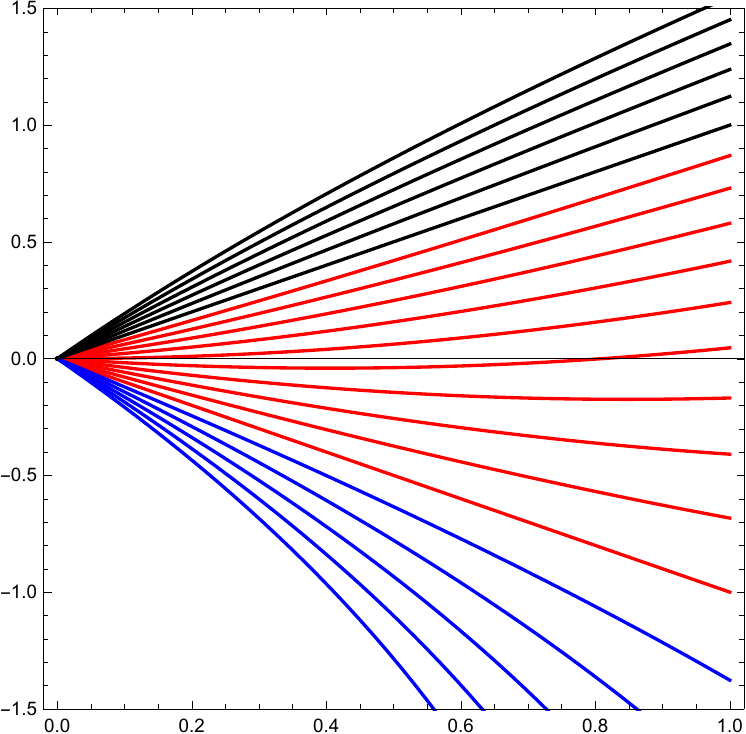}
\caption{Optimal paths $g(t)$ for the CEV model with $\beta=-\frac12$ (square-root model), for $\rho=1.0$. The paths correspond to $\log(K/S_0)$ taking values in
$\{-1.5,-1.1\}$ (blue), $\{-1,0,+1.0\}$ (red) and $\{+1.1,+1.5\}$ (black).
Each path is contained in one of the three regions in Figure~\ref{Fig:1}.}
\label{Fig:3colors}
\end{figure}

Figure~\ref{Fig:3colors} shows the optimal path $g(t)$ determined from Proposition
\ref{prop:CEV.EL},
for values of $\log(K/S_0)$ $\{-1.5,-1.1\}$ (blue), $\{-1,0,+1.0\}$ (red) and $\{+1.1,+1.5\}$ (black), in steps of 0.1.
Each path is contained in one of the three regions in Figure~\ref{Fig:1}, in agreement with the path classification analysis in Section \ref{sec:3.1}.

Figure~\ref{Fig:2} shows the asymptotic implied volatility $\sigma_{\mathrm{BBF},\rho}(K,S_0;\rho)$ vs $x = \log\frac{K}{S_0 e^\rho}$
(solid curve), in units of $\sigma/\sqrt{S_0}$. 
This is compared with the simple BBF formula (dashed curve), and with exact numerical evaluation (dots) using the analytical
results from Schroeder \cite{Schroeder}. 
The four scenarios correspond to $r=0.1$, $q=0$, and $T=\{1, 2, 5, 10\}$. 
The contribution from region 3 $(S_0 e^{-\rho} < K < S_0 e^\rho)$ is shown in red.
The agreement of the improved asymptotic result with the exact evaluations is very good. 

\begin{figure}[h]
\centering
\includegraphics[width=8cm]{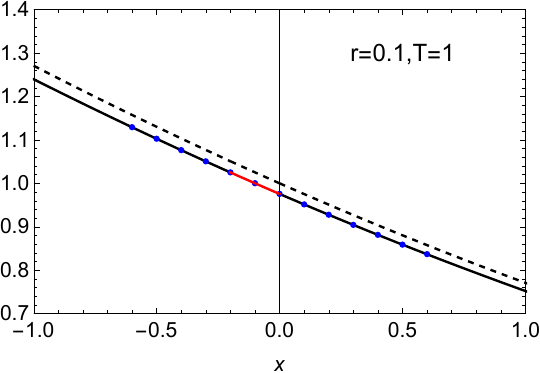}
\includegraphics[width=8cm]{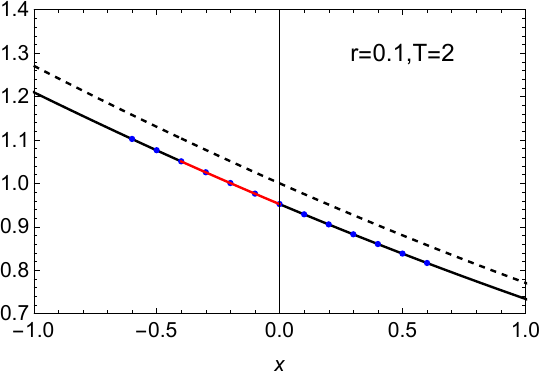}
\includegraphics[width=8cm]{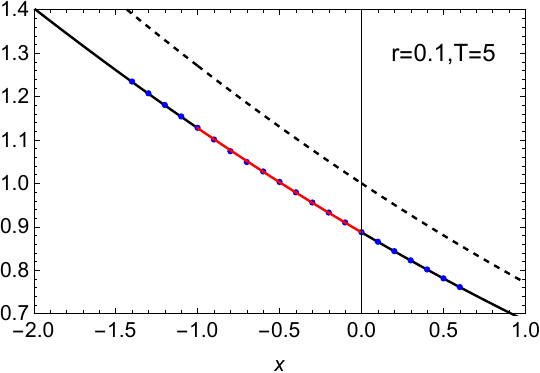}
\includegraphics[width=8cm]{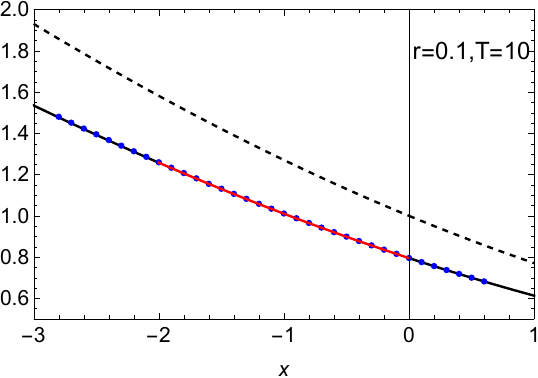}
\caption{Comparison of the improved asymptotic result $\sigma_{\mathrm{BBF},\rho}(K,S_0)$ (solid curve) 
(in units of $\frac{\sigma}{\sqrt{S_0}}$)
with the simple BBF formula $\sigma_{\mathrm{BBF}}(K,S_0)$ (dashed curve), and exact numerical evaluation (dots), for the 
$\beta=-\frac12$ model (square-root model). The contribution from region 3 is shown in red.}
\label{Fig:2}
\end{figure}

We study also the dependence on $\sigma$ by comparing the exact ATM implied volatility with the improved asymptotic result (\ref{sigATMrho}). The results are shown in Figure~\ref{Fig:3} for several choices of $r, T$. These plots show the normalized ATM implied volatility $\sigma_{ATM} := \frac{S_0}{\sigma} \sigma_{BS}(K=F(T),T)$. The improved asymptotic result becomes exact in the $\sigma \to 0$ limit. As $\sigma$ increases, the asymptotic result underestimates the exact result but the difference 
remains small for all $\sigma < 0.7$, which corresponds to ATM implied vols of about 50\%.

\begin{figure}[h]
\centering
\includegraphics[width=8cm]{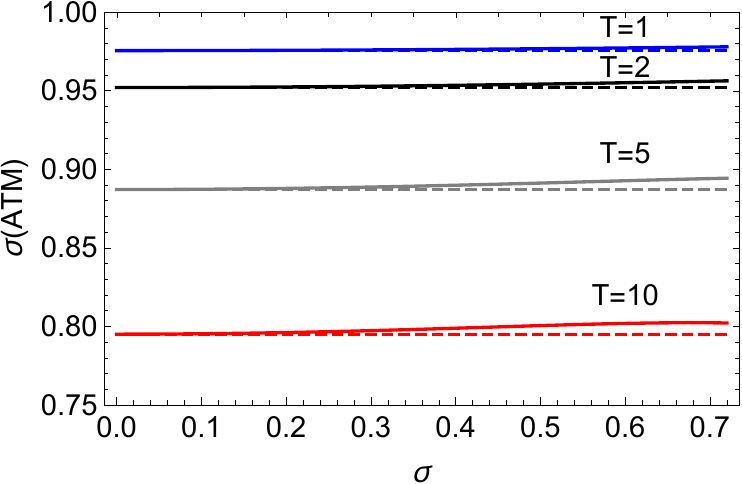}
\includegraphics[width=8cm]{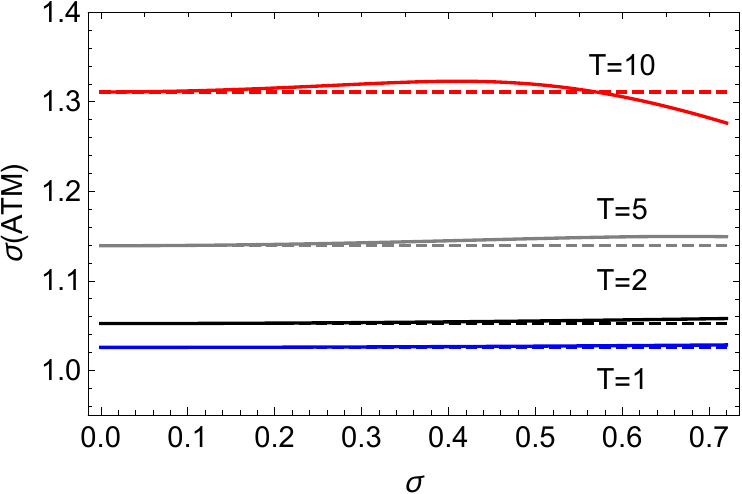}
\caption{The exact ATM implied volatility $\sigma_{BS}(K=S_0 e^{rT},T)$ (solid curves)
(in units of $\frac{\sigma}{\sqrt{S_0}}$), compared with the improved asymptotic
result $\sigma_{\mathrm{BBF},\rho}$ in (\ref{sigATMrho}) (dashed curves), vs $\sigma$. 
The simple BBF result is 1.
Parameters: $\beta=-\frac12$ model (square-root model), $r=0.1$ (left) and $r=-0.1$ (right).}
\label{Fig:3}
\end{figure}

\section*{Acknowledgements}
Lingjiong Zhu is partially supported by the grants NSF DMS-2053454, NSF DMS-2208303.

\appendix

\section{Background of Large Deviations Theory}\label{app:LD}

We give in this Appendix a few basic concepts of Large Deviations Theory which will be used in the proofs of this paper.
We refer to Dembo and Zeitouni \cite{Dembo1998} for more details on large deviations and its applications.

\begin{definition}[Large Deviation Principle]
A sequence $(P_\epsilon)_{\epsilon \in \mathbb{R}^+}$ of probability measures
on a topological space $X$ satisfies the large deviation principle with rate function $I: X \to \mathbb{R}$
if $I$ is non-negative, lower semicontinuous and for any measurable set $A$, we have
\begin{equation}
- \inf_{x\in A^o} I(x) \leq \liminf_{\epsilon\to 0} \epsilon \log P_\epsilon(A) \leq
\limsup_{\epsilon\to 0} \epsilon \log P_\epsilon(A) \leq - \inf_{x\in \bar A} I(x) \,.
\end{equation}
Here, $A^o$ is the interior of $A$ and $\bar A$ is its closure.
\end{definition}

In the proof of Theorem \ref{thm:LD} we use Varadhan's lemma. For the convenience of the
readers, we state the result as follows. 

\begin{lemma}[Varadhan's Lemma]
Suppose that $P_\epsilon$ satisfies a large deviation principle with rate function $I:X\to \mathbb{R}^+$
and let $F:X \to \mathbb{R}$ be a bounded and continuous function. Then 
\begin{equation}
\lim_{\epsilon\to 0} \epsilon \log \int_X e^{\frac{1}{\epsilon} F(x)} dP_\epsilon(x) = \sup_{x\in X} \{ F(x) - I(x) \}\,.
\end{equation}
\end{lemma}

\section{Novikov condition for the CEV model}
\label{app:1}

We give in this Appendix sufficient conditions for the finiteness of the
expectation appearing in the Novikov condition for the CEV model
\begin{equation}\label{CEV}
dS_t = \sigma S_t^{\beta+1} dW_t+ (r-q) S_t dt\,,
\end{equation}
with $-1/2 \leq \beta < 0$. For simplicity of notation we will assume $q=0$.

As in the proof of Theorem~\ref{thm:LD}, 
we aim to show that $\mathbb{E}\left[\exp\left(\frac{1}{2}\int_{0}^{T}\frac{r^{2}dt}{\sigma^2 S_{t}^{2\beta}} \right)\right] < \infty$
for sufficiently small $T>0$.
Denote the expectation to be studied as
\begin{equation}\label{Idef}
I_\beta(T) := \mathbb{E}\left[e^{\frac{r^{2}}{2\sigma^2} \int_0^T S_t^{2|\beta| } dt }\right] \,.
\end{equation}

We distinguish between the two cases $\beta = -\frac12$ and $\beta \in (-\frac12,0)$.
For both cases we prove that the expectation (\ref{Idef}) is finite, for sufficiently small $T$.

\begin{proposition}
Assume that $S_t$ follows the square root model 
$dS_t = \sigma \sqrt{S_t} dW_t + r S_t dt$. Then the expectation 
\begin{equation}\label{Isqrt}
I_{-1/2}(T) = \mathbb{E}\left[e^{\frac{r^{2}}{2\sigma^2} \int_0^T S_t dt }\right]
\end{equation}
is finite for $T< T_{exp}(r):=2/r$.
\end{proposition}

\begin{proof}
In the square root model the expectation (\ref{Isqrt}) can be computed exactly, as shown by Cox et al. \cite{Cox1985}. We will use the form of the result 
quoted in equation (4.1) in Dufresne \cite{Dufresne2001}:
\begin{equation}\label{I:s:eqn}
I(s) := \mathbb{E}\left[ e^{- s \int_0^T S_t dt }\right] = \exp\left(
- s \frac{S_0}{P} \cdot \frac{2 \sinh(P T/2)}{\cosh(P T/2) - \frac{r}{P} \sinh(P T/2)}\right)\,,
\end{equation}
with $P=\sqrt{r^2 + 2\sigma^2 s}$. 
The expectation (\ref{Isqrt}) corresponds to 
$s=-\frac{r^{2}}{2\sigma^2}$, which yields $P=0$ so that
by taking the limit $P\rightarrow 0$ in \eqref{I:s:eqn} we obtain
\begin{equation}
I_{-1/2}(T) =I\left(-\frac{r^{2}}{2\sigma^{2}}\right)
= \mathbb{E}\left[ e^{ \frac{r^{2}}{2\sigma^2} \int_0^T S_t dt }\right] = 
\exp\left(\frac{r^{2}T}{2\sigma^{2}} S_0 \frac{1}{1-\frac{rT}{2}}\right),
\end{equation}
which is finite for any $T<2/r$. 
This completes the proof.
\end{proof}

A similar result holds in the more general CEV model 
with $-\frac12 < \beta < 0$.

\begin{proposition}
Assume that the asset price $S_t$ follows the CEV model 
$dS_t = \sigma S_t^{\beta+1} dW_t+ r S_t dt$
 with
$-\frac12 < \beta < 0$. 
Then the expectation (\ref{Idef}) is finite
$I_\beta(T) < \infty$ for 
$T < - \frac{1}{2|\beta| r}\log(1 - 2|\beta |\pi)$ if $|\beta |<\frac{1}{2\pi}$ and
for all $T>0$ if $\frac{1}{2\pi} \leq |\beta | < \frac{1}{2}$.
\end{proposition}

\begin{proof}
It is well known that the process (\ref{CEV}) 
can be reduced to the square root model by a sequence of two transformations.

\textbf{Step 1.} 
Remove the drift in (\ref{CEV}) by the redefinition $S_t = x_t e^{rt}$
where $x_t$ follows the process
$dx_t = \sigma x_t^{\beta+1} e^{\beta rt} dW_t$.
The time dependent factor can be absorbed into a time redefinition as
$dx_\tau = \sigma x_\tau^{\beta+1} dW_\tau$,
with $\tau(t) = \frac{1}{2\beta r} (e^{2\beta rt} - 1 )$.

\textbf{Step 2.} Introduce $z_\tau := x_\tau^{-2\beta}$
which follows the square-root process with constant drift
\begin{equation}\label{zSDE}
dz_\tau = - 2\sigma \beta \sqrt{z_\tau} dW_\tau + \sigma^2 \beta (2\beta+1) d\tau\,.
\end{equation}

The integral in the exponent of (\ref{Idef}) becomes
\begin{equation}\label{eqn:integral}
\int_0^T S_t^{2|\beta |} dt = \int_0^T x_t^{2|\beta |} e^{2|\beta | rt} dt
= \int_0^T z_t e^{ 2 |\beta | r t } dt \,.
\end{equation}

For any $0< t\leq T$, we have $e^{ 2|\beta | r t} =\frac{1}{1 - 2 | \beta  | r \tau(t)} \leq \frac{1}{1 - 2 | \beta  | r \tau(T)}$ since $t \mapsto \tau(t)$ is a monotonically 
increasing function. Thus, the integral in \eqref{eqn:integral} is bounded from above as
\begin{equation}
\int_0^T z_t e^{2 |\beta | r t } dt \leq \frac{1}{(1 - 2 |\beta | r \tau(T))^2 } 
\int_0^{\tau(T)} z_s ds\,,\qquad\text{for any $T>0$.}
\end{equation}
Thus the expectation appearing in the Novikov condition is bounded as
\begin{equation}
\mathbb{E}\left[e^{\frac{r^{2}}{2\sigma^2} \int_0^T S_t^{2|\beta |} dt } \right]\leq
\mathbb{E}\left[e^{\frac{r^{2}}{2\sigma^2(1 - 2 |\beta | r \tau(T))^2 } \int_0^{\tau(T)} z_s ds }\right].
\end{equation}

The expectation giving the upper bound is obtained by 
replacing $\gamma \to \frac{r^{2}}{2(1-2 |\beta | r \tau(T) )^2}$ in 
Lemma~\ref{lemma:1}.
By Lemma~\ref{lemma:1}, this expectation is finite, for all $\tau(T)$ for which
$r\tau(T) < \pi(1 - 2|\beta | r\tau(T) ) < \pi$. This is equivalent to
$e^{-2|\beta | rT} > 1 - 2 |\beta | \pi$.
For $|\beta | \geq\frac{1}{2\pi}$ this holds for all $T>0$, while for 
$|\beta | <\frac{1}{2\pi}$ it holds for sufficiently small $T$.
This completes the proof.
\end{proof}

\begin{lemma}\label{lemma:1}
Suppose that $z_t$ is defined by the process
\begin{equation}\label{z:SDE}
dz_t = \sigma \sqrt{z_t} dW_t - a dt 
\end{equation}
with $a>0$ and initial condition $z_0>0$ up until the time $t_0=\inf \{ t\geq 0; z_t=0 \}$, and $z_t=0$ for all $t>t_0$.
Then we have
\begin{equation}\label{Jbound}
J(\gamma) := \mathbb{E}\left[e^{\frac{\gamma}{\sigma^2} \int_0^\tau z_t dt }\right] \leq \exp\left( \sqrt{\frac{\gamma}{2}} \frac{z_0}{\sigma^2}
\tan\left( \sqrt{\frac{\gamma}{2}} \tau \right)\right)\,,
\end{equation}
which is finite for all $\tau < \frac{\pi}{\sqrt{2\gamma}}$.
\end{lemma}

\begin{proof}
Denote $y_t$ the process defined by $dy_t = \sigma \sqrt{y_t} dW_t$ until the first time it hits zero, with absorbtion at origin,
and started at the same value as $z_t$, that is $y_0=z_0 > 0$.  

We would like to use the comparison theorem for solutions of one-dimensional SDEs 
(Theorem 1.1 in \cite{Ikeda1977}) to compare pathwise $z_t$ and $y_t$.
The comparison theorem assumes that the volatility function $\sigma(x)$ satisfies
$| \sigma(x) - \sigma(y) | \leq \rho( |x-y| ), x,y\in \mathbb{R}$, with $\rho(\xi)$ an increasing function on
$[0,\infty)$ such that $\rho(0)=0$ and $\int_0^\infty \rho(\xi)^{-2} d\xi=\infty$.
This condition is satisfied by $\sigma(x) = \sigma\sqrt{x}$ with $h(\xi) = \sqrt{\xi}$, as can be seen from the inequality
$|\sqrt{y} - \sqrt{x} | \leq \sqrt{y-x}$ which holds for any $0< x < y$. 

Application of the comparison theorem gives $z_t \leq y_t$ almost surely, which implies an inequality among the expectations
\begin{equation}\label{eqn:two:sides}
\mathbb{E}\left[e^{\frac{\gamma}{\sigma^2} \int_0^T z_t dt }\right] \leq
\mathbb{E}\left[e^{\frac{\gamma}{\sigma^2} \int_0^T y_t dt }\right] \,.
\end{equation}

The expectation on the right hand side of \eqref{eqn:two:sides} 
can be evaluated in closed form,  see e.g. equation (4.1) in Dufresne \cite{Dufresne2001} as
\begin{equation}\label{I0s}
\mathbb{E}\left[e^{-s\int_0^\tau y_s ds} \right] = \exp\left( -s \frac{z_0}{P} \tanh\left(\frac{P \tau}{2} \right)\right)\,,
\end{equation}
with $P=\sigma\sqrt{2s}$.
Taking here $s = - \frac{\gamma}{\sigma^2}$ gives $P=iQ$ with $Q=\sqrt{2\gamma}$. Substituting into (\ref{I0s}) gives the stated result (\ref{Jbound}).
This completes the proof.
\end{proof}

\section{Proof of Proposition~\ref{prop:CEV}}
\label{app:2}

First we give an explicit result for the optimal paths for the CEV model. 
They are obtained by solving the
Euler-Lagrange equation for the variational problem of Theorem~\ref{thm:LD}.
Specializing (\ref{ELg}) to the CEV model, this equation becomes
\begin{equation}\label{ELCEV}
g''(t) = S_0 e^{g(t)} \cdot \frac{\sigma'(S_0 e^{g(t)} )}{\sigma(S_0 e^{g(t)} )}
\cdot [(g'(t))^2 - \rho^2] =  |\beta | \left[\rho^2 - (g'(t))^2\right] \,.
\end{equation}

\begin{proposition}\label{prop:CEV.EL}
The solutions of the Euler-Lagrange equation for the CEV model
are different in the two regions:

i) $K \geq S_0 e^{\rho}$ (region 1) and 
   $0<K\leq S_0 e^{-\rho }$ (region 2).
In this case, the solution is
\begin{equation}\label{gsol1}
g(t) = \frac{1}{|\beta |} \log\left( y_0 + \rho - (y_0 - \rho) e^{- 2 |\beta | \rho t} \right) 
- \frac{1}{|\beta |} \log(2\rho  ) + \rho  t \,,
\end{equation}
where
\begin{equation}\label{y0sol1}
y_0 =  \rho \frac{2e^{|\beta | (\log(K/S_0)- \rho ) } - (1 + e^{-2|\beta| \rho } )}
{1 - e^{-2 \rho  |\beta |}} \,.
\end{equation}

ii)  $S_0 e^{- \rho }< K < S_0 e^{ \rho }$ (region 3).  
In this case, the solution is 
\begin{equation}\label{gsol2}
g(t) = \frac{1}{|\beta |} \log\left\{ \sqrt{1-(y_0/\rho)^2} \cosh
\left( |\beta | \rho t + \arctanh (y_0/\rho) \right) \right\} \,,
\end{equation}
with 
\begin{equation}\label{y0sol2}
y_0 = \rho \frac{e^{|\beta | \log(K/S_0)} - \cosh( |\beta | \rho)}{\sinh ( |\beta | \rho) }\,.
\end{equation}
\end{proposition}

\begin{proof}
The Euler-Lagrange equation (\ref{ELCEV}) is a first order ODE for $y(t) := g'(t)$
\begin{equation}\label{yeq}
y'(t) = |\beta | (\rho^2 - y^2(t) )\,,
\end{equation}
with initial condition $y(0) = y_0$. 

i) For this case $|y_0 | \geq \rho$ and we write the equation for $y(t)$ as
\begin{equation}
\frac{dy}{y^2-\rho^2} = - |\beta | dt\,,
\end{equation}
or
\begin{equation}
\frac{1}{2\rho} \Big( \frac{dy}{y-\rho} - \frac{dy}{y+\rho } \Big) = - |\beta | dt \,.
\end{equation}
Integration gives
\begin{equation}
\frac{y(t) - \rho }{y(t) + \rho } = \frac{y_0 - \rho }{y_0 + \rho } e^{-2|\beta | \rho t}\,.
\end{equation}
Solving for $y(t)$ this gives
\begin{equation}\label{ysol1}
y(t) = \rho \frac{y_0 + \rho + (y_0 - \rho ) e^{-2|\beta | \rho t} }
{y_0 + \rho - (y_0 - \rho ) e^{-2|\beta | \rho  t}} \,.
\end{equation}
Integrating the equation for $g(t)$ with initial condition $g(0)=0$ gives
(\ref{gsol1}).
The initial condition $y_0 = g'(0)$ is determined from $g(1) = \log\frac{K}{S_0}$,
which gives (\ref{y0sol1}). 

ii) For this case $|y_0 | < |\rho|$ and the equation (\ref{yeq}) is written as
\begin{equation}
\frac{dy}{\rho^2 - y^2} = |\beta | dt\,,
\end{equation}
which can be integrated as
\begin{equation}
\frac{1}{\rho} \left( \arctanh(y/\rho) - \arctanh(y_0/\rho ) \right) = |\beta | t\,. 
\end{equation}
This gives
\begin{equation}\label{ysol1}
y(t) = \rho \tanh [|\beta | \rho t + \arctanh (y_0/\rho) ]\,.
\end{equation}
Integration of this equation with initial condition $g(0)=0$ gives (\ref{gsol2}).

The initial condition $y_0 = g'(0)$ is determined from $g(1) = \log\frac{K}{S_0}$,
which gives the equation
\begin{equation}
\cosh(|\beta | \rho t + \arctanh(y_0/\rho) ) = \frac{1}{\sqrt{1-y_0^2/\rho^2} } 
e^{|\beta | \log(K/S_0)}\,.
\end{equation}

Expanding the expression on the left hand side gives
\begin{equation}
\cosh(|\beta | \rho ) + \frac{y_0}{\rho} \sinh( |\beta | \rho ) = e^{|\beta | x}\,,
\end{equation}
which yields the result (\ref{y0sol2}).
\end{proof}

Now we are in a position to prove Proposition~\ref{prop:CEV} in the main text. 

\begin{proof}[Proof of Proposition~\ref{prop:CEV}]
The rate function is
evaluated by direct integration from the result
\begin{equation}\label{IsolA}
I(K,S_0) = \frac12 \int_0^1 \frac{(g'(t) - \rho\theta)^2}{\sigma^2(S_0 e^{g(t)} )} dt\,.
\end{equation}

i) For this case we use the solution (\ref{gsol1}) for $g(t)$. 
The factors in the integrand of (\ref{IsolA}) are evaluated as follows.
The denominator is
\begin{equation}
\sigma^2\left(S_0 e^{g(t)} \right) = \sigma^2 S_0^{2\beta} e^{2\beta g(t)} = 
\sigma^2 S_0^{2\beta} \frac{(2\rho)^2}{(y_0 + \rho - (y_0 -\rho) e^{-2\rho |\beta | t} )^2}
e^{-2\rho |\beta | t}\,,
\end{equation}
and the numerator is
\begin{equation}
y(t) - \rho\theta = 
\frac{2\rho (y_0 - \rho\theta)}{y_0 + \rho - (y_0 -\rho) e^{-2\rho |\beta | t}} e^{-\rho |\beta |(1+\theta) t}\,.
\end{equation}

Collecting all factors, we get
\begin{equation}\label{aux1}
I(K,S_0) = \frac{S_0^{2|\beta|}}{2\sigma^2 } (y_0-\rho\theta)^2 
\int_0^1 e^{-2|\beta|\rho \theta t} dt =
\frac{S_0^{2|\beta|}}{\sigma^2 } (y_0 - \rho \theta)^2 \cdot \frac{1 - e^{-2 |\beta | \rho \theta} }{4\theta |\beta | \rho} \,.
\end{equation}
Furthermore, from (\ref{y0sol1}) we have
\begin{equation}
y_0 - \rho \theta = \frac{2\rho}{1 - e^{-2\rho |\beta |} } e^{-|\beta | \rho (1-\theta) } \left( e^{|\beta | x } - 1 \right)\,,
\end{equation}
with $x = \log\frac{K}{S_0} - \rho \theta$ the log-moneyness.

Substituting into (\ref{aux1}) we get
\begin{equation}
I(K,S_0) = \frac{S_0^{2|\beta|}}{\sigma^2 }
\frac{\rho}{\theta |\beta |} \cdot
\frac{1 - e^{-2\theta |\beta | \rho} }{(1-e^{-2\rho |\beta |} )^2} \cdot e^{-2|\beta | \rho (1-\theta)}
\cdot (e^{|\beta | x } - 1)^2\,.
\end{equation}

This reproduces the quoted result (\ref{ICEVsol1}) for $\theta = \pm 1$.

ii) The proof proceeds in a similar way to case (i), starting with the solution 
(\ref{gsol2}) for $g(t)$.

The denominator in (\ref{IsolA}) is evaluated as
\begin{equation}
\sigma^2(S_0 e^{g(t)} ) = \frac{\sigma^2}{S_0^{2|\beta |} }\cdot \frac{1}{1 - y_0^2/\rho^2} \cdot
\frac{1}{\cosh^2 w(t)}\,,
\end{equation}
with $w(t) := |\beta | \rho t + \arctanh (y_0/\rho)$.

Substituting this expression and (\ref{y0sol2}) into (\ref{IsolA}), the integrand becomes 
\begin{eqnarray}
I(K,S_0) &=& \frac12 \int_0^1 \frac{(g'(t)-\rho\theta)^2}{\sigma^2(S_0 e^{g(t)} )} dt \\
&=&
\frac12 \rho^2 \frac{S_0^{2|\beta |} }{\sigma^2} \left(1- \frac{y_0^2}{\rho^2} \right)
\cdot \int_0^1  (\tanh w(t) - \theta)^2 \cosh^2 w(t) dt \nonumber \\
&=& \frac{S_0^{2|\beta |} }{2\sigma^2} \rho^2 
\left( 1- \frac{y_0^2}{\rho^2} \right) \int_0^1 e^{-2 |\beta | \rho\theta t - 2
\arctanh (y_0/\rho) } dt\,. \nonumber
\end{eqnarray}
Performing the $t$ integration reproduces (\ref{ICEVsol2}).

\end{proof}

\bibliographystyle{plain}
\bibliography{ShortMaturityRhoLimit}


\end{document}